\documentclass[12pt]{article}
\usepackage{amsfonts}
\usepackage{amssymb}
\usepackage{graphicx}
\usepackage{amsmath}
\usepackage{color}
\usepackage{ulem}
\usepackage{pstricks}
\usepackage{pst-node}

\newtheorem{theorem}{Theorem} [section]

\newtheorem{definition}[theorem]{Definition}
\newtheorem{example}[theorem]{Example}

\newtheorem{lemma}[theorem]{Lemma}

\newtheorem{proposition}[theorem]{Proposition}

\newenvironment{proof}[1][Proof]{\textbf{#1.} }{\ \rule{0.5em}{0.5em}}
\allowdisplaybreaks[1]
\begin{document}

\author{Anjeza Bekolli\thanks{%
Department of Mathematics and Informatics, Faculty of Economy and Agribusiness, Agricultural University of Tirana, Albania. e-mail: abekolli@ubt.edu.al}, Luis A. Guardiola\thanks{%
Departamento de Métodos Cuantitativos para la Economía y
la Empresa, Universidad de Murcia, Murcia 30100, Spain. e-mail:
guardiola@um.es} $^{,} $\thanks{%
Corresponding author.} \ and Ana Meca\thanks{I. U. Centro de Investigación Operativa, Universidad Miguel Hernández de Elche, 03202 Elche, Spain.
e-mail: ana.meca@umh.es}}
\title{Profit allocation in agricultural supply chains: exploring the nexus of cooperation and compensation}
\maketitle

\begin{abstract}
In this paper, we focus on decentralized agricultural supply chains consisting of multiple non-competing distributors satisfying the demand of their respective markets. These distributors source a single product from a farmer through an agricultural cooperative, operating in a single period. The agents have the ability to coordinate their actions to maximize their profits, and we use cooperative game theory to analyze cooperation among them. The distributors can engage in joint ordering, increasing their order size, which leads to a decrease in the price per kilogram. Additionally, distributors have the opportunity to cooperate with the farmer, securing a reduced price per kilogram at the cost price, while compensating the farmer for any kilograms not acquired in the cooperation agreement. We introduce multidistributor-farmer games and we prove that all the agents have incentives to cooperate. We demonstrate the existence of stable allocations, where no subgroup of agents can be better off by separating. Moreover, we propose and characterize a distribution of the total profit that justly compensates the contribution of the farmer in any group of distributors. 
Finally, we explore the conditions under which the farmer can be compensated in order to maximize their revenues when cooperating with all players. 


\end{abstract}

\newpage

\section{Introduction}

The agricultural sector is one of the most important sectors in a country. The revenue generated by this sector is significantly impacted by the effective management of every component in it. The supply chain is responsible for the process of transforming an agricultural product from a raw material into a finished good that is ready for market consumption.

In recent years, there has been a significant increase in scientific papers that analyze and study supply chain management in the agricultural sector. A detailed review related to the supply chain in agriculture for a 10-year period can be found in the material of Khandelwal et al. (2021). A more efficient management of the supply chain implies a rise in overall revenue and a decrease in costs, both financially and in terms of time. One of the key characteristics of the supply chain is the need for cooperation and coordination among its members to achieve optimal results. Collaboration within a supply chain can take various forms, such as pooling inputs among suppliers or retailers engaging in joint orders to increase quantity demand and reduce purchase costs. These collaborations are frequently implemented in real-world situations. However, when members of the supply chain collaborate, one of the biggest problems faced by actors in the chain is how to share the benefits. The results may be cost savings or profit revenue, and in both cases, sharing becomes a key concern for the efficient administration of the collaboration. In reality, agents frequently take a subjective attitude toward profit distribution in an effort to optimize their own interests. This method frequently results in disagreements, which ultimately prevent cooperation from succeeding.

Cooperative game theory, one of the main branches of game theory, is employed to distribute the outcomes of cooperation among the members of a supply chain. In Chen et al. (2019), cooperative game theory is employed to study the technological spillover, power dynamics, and coordination of firms' environmentally-friendly actions within a supply chain. A normative framework based on cooperative game theory is proposed in Ciardiello et al. (2020) to study pollution responsibility allocation in multi-tier supply chains. This model focuses on a linear supply chain and employs three responsibility principles (Upstream, Downstream, and Local Responsibility) and developing some associated pollution responsibility allocation rules. The core of inventory centralization games is studied in Hartman et al. (2000) and Guardiola et al. (2021). In Nagarajan and Sosic (2008), Meca and Timmer (2008) and Rzeczycki (2022) we find different surveys with some applications of cooperative game theory to supply chain management. This methodology goes beyond a simple percentage-based profit-sharing approach and is founded on principles that are agreeable to all parties involved, ensuring stability in cooperation. The present study draws inspiration from Guardiola et al. (2007), which examines and analyzes cooperation and benefit distribution in a supply chain comprising a supplier and multiple retailers through the application of cooperative game theory. 

In this paper, we  make use of cooperative game theory to address a supply chain scenario involving a single period and a single agricultural product. In this specific supply chain structure, distributors place orders to farmer for a certain quantity of the product. The farmer, represented by an agricultural cooperative acting on their behalf at no cost, fulfills the demands of multiple distributors. This agricultural cooperative serves as an intermediary between the farmer and the distributors. There are several research papers that emphasize the importance of cooperatives in agricultural supply chains. For instance, in Zi et al. (2021), the authors examine a three-level supply chain involving farmers, cooperatives, and retailers in the context of fresh agricultural products. The study demonstrates that implementing a suitable relational contract can enhance product freshness and increase supply chain profits, although it may not guarantee complete stability. Another example is found in Agbo et al. (2015), which explores a model involving a marketing cooperative and direct selling, where numerous farmers are members of the agricultural marketing cooperative. The authors conclude that direct selling can foster "healthy emulation" among farmers, resulting in increased production that benefits the cooperative.

In our study, the supply chain consists of a farmer and several distributors. We consider two possible modes of operation. Firstly, distributors can collaborate by consolidating their individual orders into a single order through the agricultural cooperative. By combining their orders, they can benefit from a reduced purchase price from the farmer due to the larger quantity. In this scenario, the agricultural cooperative simply forwards the combined order to the farmer without providing any additional information. The second option is for the distributors and the farmer to collaborate together. In this case, the agricultural cooperative facilitates an agreement that binds both parties, requiring distributors to compensate the farmer for any unsold goods at a reasonable price. This price is determined based on a logical basis that ensures profitability for both parties. Additionally, under this agreement, the profits generated from the product's sale in the market are shared between the distributors and the farmer. Moreover, we assume that the quantity of goods ordered by the distributors from the farmer never exceeds the total quantity produced by the farmer. This ensures that there is no depleting the harvest in the supply of the product while also providing ongoing compensation for the amount of product remaining in stock. 

Cooperation plays a crucial role in the success of supply chains, prompting us to employ cooperative game theory for their analysis and study. We develop a cooperative transferable utility game tailored specifically to this context, involving players such as distributors and farmer. This game is referred to as the multidistributor-farmer games (MDF-games). In this game, the value of a coalition composed by all the distributors and the farmer is equal to the highest profit they can generate together. Notably, when collaboration involves all members (distributors and the farmer), the overall benefit surpasses any number of benefits that could be attained through smaller coalitions. In order to encourage the farmer and distributors' engagement in the biggest coalition, it is crucial to assure them that such allocations, which benefit all parties, exist. In the paper, we demonstrate that MDF-games are balanced, thereby ensuring the members of the cooperation that stable allocations exist, which satisfy the interests of all parties. To ensure a satisfactory distribution among the members of the supply chain, particularly for the farmer, we propose and characterize a specific allocation of the total surplus obtained from the collaboration of all agents. This allocation is known as the farmer compensation allocation (henceforth FC-allocation), which guarantees coalitional stability and provides a satisfactory distribution of the surplus, addressing the interests and welfare of all parties involved, with a particular focus on compensating the farmer. By implementing the farmer compensation allocation, we strive to establish a stable and balanced arrangement that promotes cooperation and compensation within the supply chain.

In summary, our contributions expand upon the findings presented in Guardiola et al. (2007) by incorporating two significant realistic aspects:

\begin{itemize}

\item We take into account the limited production capacity of the farmer. Additionally, in our framework, cooperating distributors are obliged to compensate the farmer for the portion of the production that remains unsold.

\item Differing from the allocation proposed in Guardiola et al. (2007), the FC-allocation provides distinct compensation to the farmer for each distributor, considering their respective shares of the marginal contribution that the farmer brings to the coalitions in which they participate. Moreover, we present a characterization of this allocation based on only three essential properties.
\end{itemize}

This paper is structured into five sections. In Section 2, we begin by introducing definitions and notations in cooperative game theory. Section 3 presents the cooperative model involving a farmer and multiple distributors, associated with the corresponding multidistributor-farmer game (MDF-game). We proceed to analyze the class of MDF-games and establish the stability (in the sense of the core) of the altruistic allocation, which distributes all profits among the distributors without allocating any to the farmer. In Section 4, we shift our focus to an alternative profit allocation method that compensates the farmer for their contribution to the overall profit increase. We propose and characterize this allocation as the farmer compensation allocation. Furthermore, in Section 5, we propose another allocation, known as the minimal proportional compensation allocation, for situations where the farmer does not achieve the maximum revenue by cooperating with all the distributors. Our aim is to encourage the farmer to cooperate with all of them. To conclude the section, we provide a necessary and sufficient condition for this allocation to be coalitionally stable. Finally, in Section 6, we present a comprehensive summary of the outcomes and contributions of our paper, along with suggestions for future research directions.

\section{Preliminaries cooperative game theory}

A cooperative game with transferable utility (TU-game) consists of a set of players $N=\{1,2,3...n\}$ and the characteristic function $v\text{, which corresponds}$ to each subset of the
set $N$ with a number from the set of real numbers. The subsets formed by
the set $N$ are called coalitions, which are denoted by $S$. Formally the
characteristic function is an application $v :2^{N} \longrightarrow \mathbb{R}$ such that 
$v(\emptyset )=0$. The value $v(S)$ of the characteristic function measures  the maximum benefit that the members of the coalition $S$ can achieve by cooperating together. The coalition formed by all agents, $N$, it is called the grand coalition.

One of the main questions that the cooperative game theory studies is how to distribute the profit generated by the grand coalition, once it has been formed. 
This distribution is done through the so-called
allocations, which are represented by a vector $x \in \mathbb{R}^{n}$, where $n$ is the number of elements in the set of players $N$. The class of superadditive games is very interesting, and motivates the formation of the grand coalition as it ensures the greatest profit for the coalition. Formally, a TU-game $(N,v)$ is superadditive when for every two coalitions $S,T \subseteq N$ such that $S \cap T =\emptyset$ it is hold that $v(S\cup T)\geq v(S)+v(T).$ Additionally, TU-games in which the profit is higher the larger the coalition is are called  
strictly increasing monotone games. This is equivalent to: $v(T)\geq v(S)$, for all $S \subseteq T \subseteq N$.

Cooperative game theory offers a variety of solutions for distributing the profits generated by cooperation. There are two types of solutions: set solutions and point solutions. Set solutions are based on eliminating all those allocations that do not satisfy certain conditions, or in other words, keeping the set of allocations that do satisfy them. Point solutions are those obtained by means of an axiomatic characterization, that is to say, they are the only allocation that satisfies certain properties. 

The most important set solution is the core of a TU-game. It consists of all efficient allocations that are coalitionally stable, i.e., no coalition has an incentive to leave the grand coalition without worsening its profit. Formally, 

$$Core(N,v)=\left\{ x\in 
\mathbb{R} ^{N}:\sum_{i\in N}x_{i}=v(N)\text{ and, for all }S\subset N, \sum_{i\in S}x_{i}\geq
v(S) \right\} .$$ 

The result obtained by Bondareva (1963) and Shapley (1967) provides a necessary and sufficient condition that the core of a TU-game is not empty. Specifically, one of the most important theorems of cooperative game theory states that a TU-game has a non-empty core if and only if it is balanced. 


A point solution $\varphi$ refers to a function that, for each TU-game $(N,v)$, determines an allocation of $v(N)$. Formally, we have $\varphi:G^{N}\longrightarrow \mathbb{R}^{N}$, where $G^{N}$ denotes the class of all TU-games with player set $N$, and $\varphi_{i}(v)$ represents the profit assigned to player $i\in N$ in the game $v\in G^{N}$. Therefore, $\varphi(v)=(\varphi_{i}(v))_{i \in N}$ is a profit vector or allocation of $v(N)$. For a comprehensive overview of cooperative game theory, we recommend referring to González-Dıaz et al. (2010).

\section{Cooperation in multidistributor-farmer supply chains}

We consider a single-period, single-product agricultural supply chain. There is a farmer, represented by an agricultural cooperative acting on his behalf at no cost, who replenishes the demands of
multiple distributors. Let $N=\{1,...,n\}$ be the set of distributors, and we denote the farmer by $0$. Consequently, $N_{0}=N\cup \{0\}$ represents the set of all agents involved in the chain. In order to describe the model, we make the following natural assumptions:

\begin{itemize}
\item $Q$ represents the total harvested quantity by the farmer, measured in
kilograms.

\item $C$ denotes the total cost of the farmer's harvest.

\item Purchasing cost function $b:\mathbb{[}0,Q]\longrightarrow \mathbb{(}0,+\infty ]$ 
such that $b(q)$ represents the cost per kilogram when a total quantity $q>0$ is bought.
It is decreasing, continuous and positive function, which satisfies that $%
b(q)\cdot q$ is non-decreasing and $b(Q)>\frac{C}{Q}.$

\item Transport cost functions. For any distributor $i\in N$, $t_{i}:\mathbb{[}0,+\infty ]\longrightarrow \mathbb{(}0,+\infty ]$ such that $t_{i}(q)$ represents the transportation costs per kilogram incurred by the distributor when transporting a quantity $q>0$ of the product to the point of sale. It is decreasing, continuous and positive function, which satisfies
that $t_{i}(q)\cdot q$ is non-decreasing.

\item Distributor price function. For any distributor $i\in N$, $p_{i}:\mathbb{[}0,+\infty ]\longrightarrow \mathbb{R}$ such that $p_{i}(q)$ represents the price function per kilogram of the product for the distributor, given a quantity $q>0$. This function provides the recommended
selling price for the quantity $q$ in the market. It is decreasing and
continuous function, which satisfies that $p_{i}(q)\cdot q$ is non-decreasing
and $p_{i}(0)>t_{i}(0)+b(0).$

\item  $\overline{b}$ represents the compensation cost that distributors would have to pay to the farmer for the unpurchased quantities of the harvest if they cooperate. 
\end{itemize}

We adopt the following notation for the decision variables in the model:

\begin{itemize}
\item $q_{i}$ denotes the quantity of product ordered by the distributor $i\in N$
from the farmer.

\item $q_{S}:=\sum_{i\in S}q_{i},$ denotes the total order size by a
coalition $S\subseteq N$ of distributors to the farmer$.$
\end{itemize}

\medskip A multidistributor-farmer situation (MDF-situation) is a tuple $%
(N_{0},Q,C,b,T,P,\overline{b})$ where $T=(t_{1},...,t_{n})$ and $%
P=(p_{1},...,p_{n})$. Within this supply chain, the distributors submit
their orders to the agricultural cooperative, which plays the role of an
intermediary. Prior to placing the order, the cooperative provides them with
information regarding the harvested quantity of product $Q$, the purchasing cost
function $b(q)$ and the compensation cost $\overline{b}$. Figure 1 graphically illustrates a
MDF-situation.

\medskip

\begin{figure}[htbp]
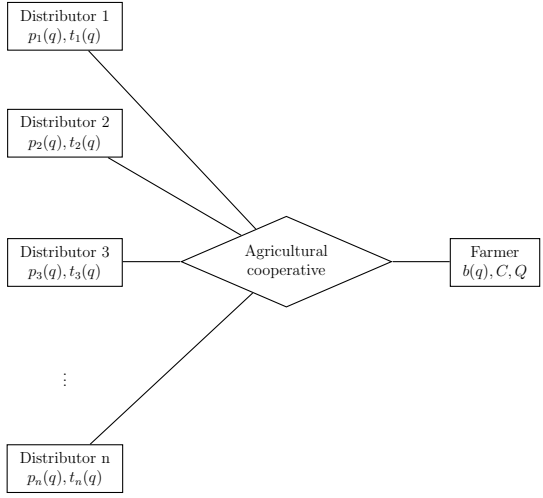

\vspace*{-0.5cm}
   \begin{center}
   \resizebox{.48\linewidth}{!}{
  \begin{psmatrix}
    & [name=dis1] \psframebox{\begin{tabular}{c}  Distributor 1 \\ $p_{1}(q),t_{1}(q)$ \end{tabular}} & &  \\
   & [name=dis2] \psframebox{\begin{tabular}{c}  Distributor 2 \\$p_{2}(q),t_{2}(q)$ \end{tabular}} & &  \\
   & [name=dis3] \psframebox{\begin{tabular}{c}  Distributor 3 \\ $p_{3}(q),t_{3}(q)$ \end{tabular}} &  [name=cop,mnode=dia] \begin{tabular}{c}Agricultural \\ cooperative \end{tabular}& [name=far] \psframebox{\begin{tabular}{c}  Farmer \\ $b(q),C,Q$ \end{tabular}} \\
& \vdots & &  \\
 & [name=disn] \psframebox{\begin{tabular}{c}  Distributor n \\ $p_{n}(q),t_{n}(q)$ \end{tabular}} & &  \\
        \ncline{-}{dis1}{cop}
	\ncline{-}{dis2}{cop}
 	\ncline{-}{dis3}{cop}
	\ncline{-}{disn}{cop}
	\ncline{-}{far}{cop}
  \end{psmatrix}}
\end{center}
\vspace*{-1.5cm}
    \caption{MDF-situations \label{figure}}
\end{figure}

We consider that agents involved in a MDF-situation may cooperate in two
different ways. One possible scenario is cooperation that excludes the
farmer, and involves distributors coordinating their actions. First,
distributors can cooperate by placing joint orders to the farmer through the
agricultural cooperative, which results in a discounted unit price and
ensures profitability in their cooperation. The cooperative only transmits
the requested order to the farmer without giving any further information.
Then, the joint profit of a coalition $S$ equals:%
\begin{eqnarray*}
\max &&\sum_{i\in S}\Pi _{i}^{S}\left( q\right) :=\sum_{i\in S}\left[ \left(
p_{i}(q_{i})-t_{i}(q_{i})-b(q_{S})\right) \cdot q_{i}\right] \\
\mbox{s.t.} &&q\in \mathbb{Q}^{S},
\end{eqnarray*}%
where 
\begin{equation*}
\mathbb{Q}^{S}:=\left\{ q\in \mathbb{R}_{+}^{\left\vert S\right\vert
}:q_{S}\leq Q\right\} .
\end{equation*}

Note that $\Pi _{i}^{S}:\mathbb{R}^{\left\vert S\right\vert }\longrightarrow 
\mathbb{R}$ defined as $\Pi _{i}^{S}\left( q\right) :=$ $\left(
p_{i}(q_{i})-t_{i}(q_{i})-b(q_{S})\right) \cdot q_{i}$ is the profit for distributor $i$ 
when ordering a quantity $q_{i}$ jointly with the coalition $S$.
It is a function of several variables $(q_{1},q_{2},...,q_{i},...,q_{\left\vert S\right\vert }).$
Moreover, $\mathbb{Q}^{S}$ is a compact set and the objective function $\sum_{i\in
S}\Pi _{i}^{S}\left( q\right) $ is continuous on $\mathbb{Q}^{S}.$ Let $%
q^{S}=\left( q_{i}^{S}\right) _{i\in S}$ denote the optimal order size for
coalition $S$, where $q_{i}^{S}$ is the optimal order size for distributor $%
i\in S$. The reader also may notice that
there may be feasible order sizes that cause distributors
losses. However, these will never be optimal, since we are guaranteed to
have order sizes for which there is a profit ($p_{i}(0)>t_{i}(0)+b(0)$ for
all $i\in N$).

The other possible scenario is cooperation that includes the farmer, which
involves the farmer $0$ collaborating with a group of distributor $S$. This 
coalition is denoted by $S_{0}:=S \cup \{0\}.$ In this scenario, the agricultural cooperative obliges the distributors to sign
an exclusivity contract with the farmer, so that the distributors must pay
compensation to the farmer, denoted as $\overline{b}$, for the unpurchased
quantities of the harvest $Q-q_{S}$. Each distributor $i\in S$ provides part
of this compensation proportionally to the amount of order that was requested $%
\frac{q_{i}}{q_{S}}.$ Moreover, they are expected to share with the farmer the total profits
from their sales to the public across different markets. On the other hand, the
farmer commits to enabling the distributors to acquire the desired quantity
at the cost price of $\frac{C}{Q}$. Therefore, the joint profit of a coalition $%
S_{0}$ equals:%
\begin{eqnarray*}
\max &&\sum_{i\in S}\Lambda _{i}^{S}\left( q\right) :=\sum_{i\in S}\left[
\left( p_{i}(q_{i})-t_{i}(q_{i})-\frac{C}{Q}-\frac{\overline{b}\cdot
(Q-q_{S})}{q_{S}}\right) \cdot q_{i}\right] \\
\mbox{s.t.} &&q\in \mathbb{Q}^{S},
\end{eqnarray*}

As in the previous case, $\sum_{i\in S}\Lambda _{i}^{S}\left( q\right) $ is
continuous on $\mathbb{Q}^{S}$ for all coalition $S\subseteq N.$ Let $%
q^{S_{0}}=\left( q_{i}^{S_{0}}\right) _{i\in S}$ denote the optimal order
size for coalition $S$ when they cooperate with the farmer. 

Despite the fact that these optimization problems aim to maximize the profit
from selling to the public across different markets without considering the
farmer's revenue, we can calculate it for each coalition $S\subseteq N\ $ as 
$r(S):=b(q_{S}^{S})\cdot q_{S}^{S}$ and $r(S_{0}):=\frac{C}{Q}\cdot q_{S}^{S_{0}}+\overline{b}\cdot (Q-q_{S}^{S_{0}})$ just in case the farmer cooperates with the coalition $S$.

Our study focuses on those MDF-situations that satisfy the following two additional assumptions:

\begin{itemize}
\item[1.] Sustainable Compensation (SC): $0<\overline{%
b}\leq \underset{S\subseteq N}{\min }\left\{ \left( b(q_{S}^{S})-\frac{C}{Q}%
\right) \cdot \frac{q_{S}^{S}}{Q-q_{S}^{S}}\right\} .$

\item[2.] No Depleting the Harvest (NDH): $q_{S}^{S},q_{S}^{S_{0}}<Q$ for all $S\subseteq N.$
\end{itemize}

The first assumption guarantees that the compensation cost cannot exceed 
the increase in revenue per unsold kilo for any coalition of distributors.
The second assumption ensures that the distributors will never reach the harvested
quantity. If NDH is not satisfied, certain distributors within the coalition would
need to reduce their orders, leading to an uneven distribution where the
distributor benefiting the most would order more. This case poses a more
intricate analysis and will be proposed as a potential avenue for future
research.
\medskip

The following lemma shows the principal properties of the profit function
for the distributors.

\begin{lemma}
\label{lemma}Let $(N_{0},Q,C,b,T,P,\overline{b})$ be an MDF-situation
with SC and NDH. Then, for all
distributor $i\in S\subseteq N$ it is satisfy that:

\begin{itemize}
\item[(i)] $\Lambda _{i}^{S}\left( q^{S_{0}}\right) \geq \Lambda
_{i}^{S}\left( q\right) $ for any $q\in \mathbb{Q}^{S}$;

\item[(ii)] $\Lambda _{i}^{N}\left( q\right) \geq \Lambda _{i}^{S}\left( 
\overline{q}\right) $ for any $q\in \mathbb{Q}^{N}$ and $\overline{q}\in 
\mathbb{Q}^{S}$ such that $\overline{q}_{k}=q_{k}$ for all $k\in S$;

\item[(iii)] $\Lambda _{i}^{N}\left( q^{N_{0}}\right) \geq \Lambda
_{i}^{S}\left( q^{S_{0}}\right) $;

\item[(iv)] $\Lambda _{i}^{S}\left( q^{S}\right) \geq \Pi _{i}^{S}\left(
q^{S}\right) $;

\item[(v)] $\Lambda _{i}^{N}\left( q^{N_{0}}\right) \geq \Pi _{i}^{S}\left(
q^{S}\right) $;

\item[(vi)] $q_{i}^{N_{0}}=q_{i}^{S_{0}}$.
\end{itemize}
\end{lemma}
\begin{proof}

(i) Consider the objective function

\begin{eqnarray*}
\sum_{i\in S}\Lambda _{i}^{S}\left( q\right) &=&\sum_{i\in S}\left[ \left(
p_{i}(q_{i})-t_{i}(q_{i})-\frac{C}{Q}-\frac{\overline{b}\cdot (Q-q_{S})}{%
q_{S}}\right) \cdot q_{i}\right] \\
&=&\sum_{i\in S}\left[ \left( p_{i}(q_{i})-t_{i}(q_{i})-\frac{C}{Q}\right)
\cdot q_{i}\right] -\overline{b}\cdot (Q-q_{S}) \\
&=&\sum_{i\in S}\left[ \left( p_{i}(q_{i})-t_{i}(q_{i})-\frac{C}{Q}\right)
\cdot q_{i}\right] -\overline{b}\cdot Q+\overline{b}\cdot q_{S} \\
&=&\sum_{i\in S}\left[ \left( p_{i}(q_{i})-t_{i}(q_{i})-(\frac{C}{Q}-%
\overline{b})\right) \cdot q_{i}\right] -\overline{b}\cdot Q  \label{fol}
\end{eqnarray*}

We may notice that it is a sum of independent terms. Therefore, the maximum of the sum is equal to the sum of the maximums of each term, given that $q_{S}^{S}$ and $q_{S}^{S_{0}}$ are both less than Q for all $S\subseteq N.$ Formally, 
\begin{equation}
\underset{q\in \mathbb{Q}^{S}}{\max }\sum_{i\in S}\Lambda _{i}^{S}\left(
q\right) =\sum_{i\in S}\left[ \underset{0\leq q_{i}\leq Q}{\max }\left\{
\left( p_{i}(q_{i})-t_{i}(q_{i})-(\frac{C}{Q}-\overline{b})\right) \cdot
q_{i}\right\} \right] -\overline{b}\cdot Q.  \label{lol}
\end{equation}%
We can conclude that, $\Lambda _{i}^{S}\left( q^{S_{0}}\right) \geq \Lambda
_{i}^{S}\left( q\right) $ for any $q\in \mathbb{Q}^{S}.$

(ii) Consider a distributor $i\in S\subseteq N$,  $q\in \mathbb{Q}^{N}$ and $\overline{q}\in 
\mathbb{Q}^{S}$ such that $\overline{q}_{k}=q_{k}$ for all $k\in S$,
\begin{eqnarray*}
\Lambda _{i}^{N}\left( q\right) &=&\left( p_{i}(q_{i})-t_{i}(q_{i})-\frac{C}{%
Q}-\frac{\overline{b}\cdot (Q-q_{N})}{q_{N}}\right) \cdot q_{i} \\
&=&\left( p_{i}(q_{i})-t_{i}(q_{i})-\frac{C}{Q}-\frac{\overline{b}\cdot
(Q-q_{N\backslash S-}q_{S})}{q_{N\backslash S+}q_{S}}\right) \cdot q_{i} \\
&\geq &\left( p_{i}(q_{i})-t_{i}(q_{i})-\frac{C}{Q}-\frac{\overline{b}\cdot
(Q-q_{S})}{q_{S}}\right) \cdot q_{i}=\Lambda _{i}^{S}\left( \overline{q}%
\right)
\end{eqnarray*}

(iii) By (i) we have that $\Lambda _{i}^{N}\left( q^{N_{0}}\right) \geq
\Lambda _{i}^{N}\left( \widehat{q}^{S_{0}}\right) $ where $\widehat{q_{k}}%
^{S_{0}}=q_{k}^{S_{0}}$ for all $k\in S$ and $\widehat{q_{k}}^{S_{0}}=0$
otherwise. Finally, by (ii), \ $\Lambda _{i}^{N}\left( \widehat{q}%
^{S_{0}}\right) \geq \Lambda _{i}^{S}\left( q^{S_{0}}\right) .$

(iv) We must prove the following inequality $\Lambda _{i}^{S}\left(
q^{S}\right) \geq \Pi _{i}^{S}\left( q^{S}\right)$, which is equivalent
to:%
\begin{eqnarray*}
\left( p_{i}(q_{i}^{S})-t_{i}(q_{i}^{S})-\frac{C}{Q}-\frac{\overline{b}\cdot
(Q-q_{S}^{S})}{q_{S}^{S}}\right) \cdot q_{i}^{S} &\geq &\left(
p_{i}(q_{i}^{S})-t_{i}(q_{i}^{S})-b(q_{S}^{S})\right) \cdot q_{i}^{S}; \\
-\frac{C}{Q}-\frac{\overline{b}\cdot (Q-q_{S}^{S})}{q_{S}^{S}} &\geq
&-b(q_{S}^{S}); \\
\frac{\overline{b}\cdot (Q-q_{S}^{S})}{q_{S}^{S}} &\leq &b(q_{S}^{S})-\frac{C%
}{Q}; \\
\overline{b} &\leq &\left( b(q_{S}^{S})-\frac{C}{Q}\right) \cdot \frac{%
q_{S}^{S}}{Q-q_{S}^{S}}.
\end{eqnarray*}

Note that the last inequality is well-defined because of $q_{S}^{S}<Q$ for all $%
S\subseteq N.$ Additionally, it is true since $\overline{b}\leq 
\underset{S\subseteq N}{\min }\left\{ \left( b(q_{S}^{S})-\frac{C}{Q}\right)
\cdot \frac{q_{S}^{S}}{Q-q_{S}^{S}}\right\} .$ Hence, $\Lambda
_{i}^{S}\left( q^{S}\right) \geq \Pi _{i}^{S}\left( q^{S}\right) .$

(v) We may notice that as $q^{N_{0}}$ is optimal for $\Lambda _{i}^{N}\left( q\right) $ it
follows that $\Lambda _{i}^{N}\left( q^{N_{0}}\right) \geq \Lambda
_{i}^{N}\left( \widehat{q}^{S}\right) $ where $\widehat{q_{k}}^{S}=q_{k}^{S}$
for all $k\in S$ and $\widehat{q_{k}}^{S}=0$ otherwise. Moreover by (ii), $\Lambda
_{i}^{N}\left( \widehat{q}^{S}\right) \geq \Lambda _{i}^{S}\left(
q^{S}\right)$. Finally, by (iv) we conclude that $\Lambda _{i}^{N}\left(
q^{N_{0}}\right) \geq \Pi _{i}^{S}\left( q^{S}\right) .$

(vi) It follows inmediatly from (\ref{lol}).\hfill 
\end{proof}

\bigskip

We are ready now to introduce the cooperative game that corresponds to an
MDF-situation.

\begin{definition}
Let $(N_{0},Q,C,b,T,P,\overline{b})$ be an MDF-situation. The corresponding MDF-game $(N_{0},v)$ is
given by 
\begin{equation*}
v(S):=\sum_{i\in S}\Pi _{i}^{S}\left( q^{S}\right)
\end{equation*}%
and%
\begin{equation*}
v(S_{0}):=\sum_{i\in S}\Lambda _{i}^{S}\left( q^{S_{0}}\right)
\end{equation*}%
for all coalitions $S\subseteq N$ and $v(\emptyset )=v(0)=0.$
\end{definition}

The following proposition shows some characteristics of MDF-games.

\begin{proposition}
Let $(N_{0},Q,C,b,T,P,\overline{b})$ be an MDF-situation with SC and NDH, and $(N_{0},v)$ be the corresponding
MDF-game. Then,

\begin{itemize}
\item[(i)] $0<v(S)\leq v(S_{0})$ for all $\emptyset \neq S\subseteq N.$

\item[(ii)] $v$ is superadditive;

\item[(iii)] $v$ is strictly increasing monotone;

\item[(iv)] $v(S_{0})=\sum_{i\in S}v(\{0,i\})+\left( \left\vert S\right\vert
-1\right) \cdot \overline{b}\cdot Q$
\end{itemize}
\end{proposition}
\begin{proof}
(i) We know that $p_{i}(0)>t_{i}(0)+b(0)$ for all $i\in N.$ Then, there exist a $\hat{q}%
\in \mathbb{Q}^{S}$ such that $\Pi _{i}^{S}\left( \hat{q}\right) >0$ for all 
$i\in S.$ Hence, $v(S)=\sum_{i\in S}\Pi _{i}^{S}\left( q^{S}\right)
\geq \sum_{i\in S}\Pi _{i}^{S}\left( \hat{q}\right) >0.$ Moreover, for each $%
i\in N$ we have that $0<v(S)=\sum_{i\in S}\Pi _{i}^{S}\left( q^{S}\right)
\leq \sum_{i\in S}\Lambda _{i}^{S}\left( q^{S}\right) $ by property (iv) of
Lemma \ref{lemma}. Finally, by property (i) of the same lemma we have that $%
\sum_{i\in S}\Lambda _{i}^{S}\left( q^{S}\right) \leq \sum_{i\in S}\Lambda
_{i}^{S}\left( q^{S_{0}}\right) =v(S_{0}).$ Hence, $0<v(S)\leq v(S_{0}).$

(ii) We can distinguish two cases. 
Case 1 (without farmer). Let $S,T\subseteq N$ be two disjoint coalitions of distributors. By
definition of the MDF-game 
\begin{equation*}
v(S)+v(T)=\sum_{i\in S}\Pi _{i}^{S}\left( q^{S}\right) +\sum_{i\in T}\Pi
_{i}^{T}\left( q^{T}\right) .
\end{equation*}%
Define the specific order quantity $\hat{q}_{i}$ for retailer $i$ in
coalition $S\cup T$ by $\hat{q}_{i}=q_{i}^{S}$ if $i\in S$ and $\hat{q}%
_{i}=q_{i}^{T}$ if $i\in T$. Then, the last expresion is equal to: 
\begin{eqnarray*}
\lefteqn{\sum_{i\in S}\left[ \left(
p_{i}(q_{i}^{S})-t_{i}(q_{i}^{S})-b(q_{S}^{S})\right) \cdot q_{i}^{S}\right]
+\sum_{i\in T}\left[ \left(
p_{i}(q_{i}^{T})-t_{i}(q_{i}^{T})-b(q_{T}^{T})\right) \cdot q_{i}^{T}\right] 
} \\
&\leq &\sum_{i\in S}\left[ \left( p_{i}(\hat{q}_{i})-t_{i}(\hat{q}_{i})-b(%
\hat{q}_{S\cup T})\right) \cdot \hat{q}_{i}\right] +\sum_{i\in T}\left[
\left( p_{i}(\hat{q}_{i})-t_{i}(\hat{q}_{i})-b(\hat{q}_{S\cup T})\right)
\cdot \hat{q}_{i}\right] \\
&=&\sum_{i\in S\cup T}\left[ \left( p_{i}(\hat{q}_{i})-t_{i}(\hat{q}_{i})-b(%
\hat{q}_{S\cup T})\right) \cdot \hat{q}_{i}\right] =\sum_{i\in S\cup T}\Pi
_{i}^{S\cup T}\left( \hat{q}\right)
\end{eqnarray*}%
since $b(\hat{q}_{S\cup T})=b(q_{S}^{S}+q_{T}^{T})\leq \min
\{b(q_{S}^{S}),b(q_{T}^{T})\}$. Finally,%
\begin{equation*}
\sum_{i\in S\cup T}\Pi _{i}^{S\cup T}\left( \hat{q}\right) \leq \sum_{i\in
S\cup T}\Pi _{i}^{S\cup T}\left( q^{S\cup T}\right) =v(S\cup T),
\end{equation*}%
since the quantities $\hat{q}$ is not necessary optimal for coalition $S\cup
T$. We can conclude that $v(S)+v(T) \leq v(S\cup T)$.

Case 2 (with farmer). Let $S,T\subseteq N$ be two disjoint coalitions of distributors. We also define the specific order quantity $\hat{q}_{i}$ for
retailer $i$ in coalition $S\cup T$ by $\hat{q}_{i}=q_{i}^{S_{0}}$ if $i\in
S $ and $\hat{q}_{i}=q_{i}^{T}$ if $i\in T.$ Then,
\begin{eqnarray*}
v(S_{0})+v(T) &=&\sum_{i\in S}\Lambda _{i}^{S}\left( q^{S_{0}}\right)
+\sum_{i\in T}\Pi _{i}^{T}\left( q^{T}\right) \\
&=&\sum_{i\in S}\left[ \left( p_{i}(q_{i}^{S_{0}})-t_{i}(q_{i}^{S_{0}})-%
\frac{C}{Q}-\frac{\overline{b}\cdot (Q-q_{S}^{S_{0}})}{q_{S}^{S_{0}}}\right)
\cdot q_{i}^{S_{0}}\right] \\
&&+\sum_{i\in T}\left[ \left(
p_{i}(q_{i}^{T})-t_{i}(q_{i}^{T})-b(q_{T}^{T})\right) \cdot q_{i}^{T}\right]
\\
&=&\sum_{i\in S}\left[ \left( p_{i}(q_{i}^{S_{0}})-t_{i}(q_{i}^{S_{0}})-(%
\frac{C}{Q}-\overline{b})\right) \cdot q_{i}^{S_{0}}\right] -\overline{b}%
\cdot Q \\
&&+\sum_{i\in T}\left[ \left(
p_{i}(q_{i}^{T})-t_{i}(q_{i}^{T})-b(q_{T}^{T})\right) \cdot q_{i}^{T}\right]
\\
&\leq &\sum_{i\in S}\left[ \left( p_{i}(q_{i}^{S_{0}})-t_{i}(q_{i}^{S_{0}})-(%
\frac{C}{Q}-\overline{b})\right) \cdot q_{i}^{S_{0}}\right] -\overline{b}%
\cdot Q \\
&&+\sum_{i\in T}\left[ \left( p_{i}(q_{i}^{T})-t_{i}(q_{i}^{T})-(\frac{C}{Q}-%
\overline{b})\right) \cdot q_{i}^{T}\right] \\
&=&\sum_{i\in S\cup T}\left[ \left( p_{i}(\hat{q}_{i})-t_{i}(\hat{q}_{i})-(%
\frac{C}{Q}-\overline{b})\right) \cdot \hat{q}_{i}\right] -\overline{b}\cdot
Q \\
&\leq &\sum_{i\in S}\left[ \left( p_{i}(q_{i}^{S_{0}\cup
T})-t_{i}(q_{i}^{S_{0}\cup T})-(\frac{C}{Q}-\overline{b})\right) \cdot
q_{i}^{S_{0}\cup T}\right] -\overline{b}\cdot Q \\
&=&\sum_{i\in S\cup T}\Lambda _{i}^{S\cup T}\left( q^{S_{0}\cup T}\right)
=v(S_{0}\cup T).
\end{eqnarray*}
We can conclude that $v(S_{0})+v(T) \leq v(S\cup T)$.

(iii) This follows immediately from (i) and (ii).

(iv) We know from (\ref{lol}) that  
\begin{eqnarray*}
v(S_{0}) &=&\sum_{i\in S}\Lambda _{i}^{S}\left( q^{S_{0}}\right) =\sum_{i\in
S}\left[ \left( p_{i}(q_{i}^{S_{0}})-t_{i}(q_{i}^{S_{0}})-(\frac{C}{Q}-%
\overline{b})\right) \cdot q_{i}^{S_{0}}\right] -\overline{b}\cdot Q \\
&=&\left( \sum_{i\in S}\underset{q_{i}\in \mathbb{Q}^{\{i\}}}{\max }\left\{
\left( p_{i}(q_{i})-t_{i}(q_{i})-(\frac{C}{Q}-\overline{b}%
)\right) \cdot q_{i}\right\} \right) -\overline{b}\cdot Q \\
&=&\sum_{i\in S}\left( v(\{0,i\})+\overline{b}\cdot Q\right) -\overline{b}%
\cdot Q=\sum_{i\in S}v(\{0,i\})+\left( \left\vert S\right\vert -1\right)
\cdot \overline{b}\cdot Q
\end{eqnarray*}\hfill
\end{proof}

Property (i) guarantees that any coalition of distributors generates
profits, and they are higher with the farmer. Properties (ii) and (iii) are essential for the feasibility of
forming a grand coalition. Property (iv) highlights that as the farmer
collaborates with a larger group of distributors, the savings increase
significantly since they no longer individually compensate him for
unfulfilled orders, but rather compensate him collectively. This represents
a key difference from the findings in Guardiola et al. (2007) (property (iv) of Lemma 4.4),
which demonstrated that games become additive when cooperating with the
supplier.
\medskip

Consider a MDF-situation $(N_{0},Q,C,b,T,P,\overline{b})$ and the associated MDF-game $(N_{0},v).$
The core of this game $(N_{0},v)$ is defined as follows:

\begin{equation*}
Core(N_{0},v)=\left\{ x\in \mathbb{R}^{\left\vert N_{0}\right\vert
}\left\vert 
\begin{array}{c}
\sum_{i\in N_{0}}x_{i}=v(N_{0}); \\ 
\sum_{i\in S}x_{i}\geq v(S)\text{ and}\sum_{i\in S_{0}}x_{i}\geq
v(S_{0}),\forall S\subseteq N%
\end{array}%
\right. \right\} .
\end{equation*}%

The main result of this section shows that the MDF-games arising
from situations with sustainable compensation and no depleting the harvest, are always balanced.

\begin{theorem}
Let $(N_{0},Q,C,b,T,P,\overline{b})$ be an MDF-situation with SC and NDH, and $(N_{0},v)$ be the corresponding
MDF-game. Then, $(N_{0},v)$ is balanced.
\end{theorem}
\begin{proof}
Define the allocation $x^{a}(v)$ by $x_{i}^{a}(v):=\Lambda _{i}^{N}\left(
q^{N_{0}}\right) $ for all $i\in N$ and $x_{0}^{a}(v)=0.$ In this allocation,
the distributors receive all the profit from cooperation with the farmer,
whereas the farmer receives nothing. Note that%
\begin{equation*}
\sum_{i\in N_{0}}x_{i}^{a}(v)=\sum_{i\in N}\Lambda _{i}^{N}\left(
q^{N_{0}}\right) =v(N_{0}).
\end{equation*}%
Hence, $x^{a}(v)$ is a efficient allocation. Next, consider a coalition $%
S\subseteq N.$ Then,%
\begin{equation*}
\sum_{i\in S}x_{i}^{a}(v)=\sum_{i\in S}\Lambda _{i}^{N}\left(
q^{N_{0}}\right) \geq \sum_{i\in S}\Pi _{i}^{S}\left( q^{S}\right) =v(S),
\end{equation*}%
where the inequality follows from property (iv) in Lemma \ref{lemma}.
Finally,

\begin{equation*}
\sum_{i\in S_{0}}x_{i}^{a}(v)=\sum_{i\in S}\Lambda _{i}^{N}\left(
q^{N_{0}}\right) \geq \sum_{i\in S}\Lambda _{i}^{S}\left( q^{S_{0}}\right)
=v(S_{0}),
\end{equation*}

where the last inequality follows from property (iii) in Lemma \ref{lemma}.
We conclude that $x^{a}(v)\in Core(N_{0},v)$, thus, the game is balanced.\hfill
\end{proof}
\bigskip

The allocation $x^{a}(v)$, which we will refer to as the altruistic
allocation, is inspired by and used in Guardiola et al. (2007) under the same name. The next
example illustrates a MDF-situation composed of two distributors and one
farmer. 

\begin{example}
\label{ejemplo 1}Consider a MDF-situation with 1 farmer and 2 distributors, 
where: $Q=3000,$ $C=3000,b(q)=5-\frac{q}{2000}$ with $q\in \lbrack 0,3000],$ 
$\overline{b}=0.2$ 
\begin{equation*}
p_{1}(q)=\left\{ 
\begin{array}{cc}
8-\frac{2q}{1000}, & 0\leq q\leq 2000, \\ 
4, & q>2000.%
\end{array}%
\right. \ \ \ ;\text{\ \ }t_{1}(q)=\left\{ 
\begin{array}{cc}
2-\frac{q}{10000}, & 0\leq q\leq 10000, \\ 
1, & q>10000.%
\end{array}%
\right.
\end{equation*}%
\begin{equation*}
p_{2}(q)=\left\{ 
\begin{array}{cc}
7-\frac{3q}{1000}, & 0\leq q\leq 1000, \\ 
3, & q>1000.%
\end{array}%
\right. \text{ \ \ };\text{\ \ }t_{2}(q)=\left\{ 
\begin{array}{cc}
1-\frac{q}{10000}, & 0\leq q\leq 5000, \\ 
0.5, & q>5000.%
\end{array}%
\right.
\end{equation*}

Solving the corresponding optimization problems for the different
coalitions, we obtain the optimal orders, the farmer's profit, and 
the characteristic function, as shown in Table \ref{table 1}.

\begin{table}[!h]
\begin{center}
\begin{tabular}{|c|c|c|c|c|c|}
$S$ & $q_{1}^{S}$ & $q_{2}^{S}$ & $r\left( S\right)$ & $r\left( S\right)-C$ & $v(S)$ \\ 
\hline
$\{1\}$ & $357.14$ & - & $1721.94$ & $-1278.06$ & $178.57$ \\ \hline
$\{2\}$ & - & $208.33$ & $1019.97$ & $-1980.03$ & $104.17$ \\ \hline
$\{0,1\}$ & $1368.42$ & - & $1694.74$ & $-1305.26$ & $2957.89$ \\ \hline
$\{0,2\}$ & - & $896.55$ & $1317.24$ & $-1682.76$ & $1731.03$ \\ \hline
$\{1,2\}$ & $466.24$ & $305.47$ & $3560.75$ & $+560.75$ & $385.85$ \\ \hline
$\{0,1,2\}$ & $1368.42$ & $896.55$ & $2411.98$ & $-588.02$ & $5288.92$ \\ \hline
\end{tabular}%
\caption{MDF-game for Example \ref{ejemplo 1}\label{table 1}}
\end{center}
\end{table}

Notice that it satisfy NDH since $q_{S}^{S},q_{S}^{S_{0}}<Q$ for all $S\subseteq N.$
Moreover, it does SC as $\overline{b}=0.2\leq \underset{S\subseteq N}{\min }\left\{ \left(
b(q_{S}^{S})-\frac{C}{Q}\right) \cdot \frac{q_{S}^{S}}{Q-q_{S}^{S}}\right\}
=\min \{0.516,0.29,1.251\}.$ The altruistic allocation is $x^{a}(v)=(0,3195.39,2093.53).$ 
\end{example}

As we can see in the above example, when players cooperate with the farmer, they
achieve a significant increase in their profits. However, the altruistic
distribution does not take into account the farmer's contribution to that
increase. Furthermore, although it is a stable allocation, it does not
consider the binding agreement that distributors sign with the agricultural
cooperative to share those benefits with the farmer. Therefore, in the next
section, we aim to find a stable allocation that compensates the farmer
for his contribution to the cooperative profit. 

\section{Cooperation and compensation to the farmer}

We now focus on finding an alternative core allocation of
profit that compensates the farmer for their contribution to the benefit of
the grand coalition. Our aim is that each distributor compensates the farmer with a share of the farmer's lower marginal contribution to the coalitions in which this distributor participates. This will result in potentially different compensation amounts from each
distributor. The reader may note that this distribution is different from the one proposed in Guardiola et al. (2007),
mgpc-solution, where each retailer contributed the same amount (the minimal
gain per capita) to compensate the supplier. 

On the basis of altruistic allocation, we propose the
following three desirable properties for a single-valued solution $\varphi $
on MDF-games $(N_{0},v)$\:

\begin{description}
\item[(EF)] Efficiency. $\sum_{i\in N_{0}}\varphi _{i}(v)=v(N_{0}).$

\item[(DR)] Distributors reduction. $\varphi _{i}(v)=x_{i}^{a}(v)-\frac{%
v(S_{0}^{i})-v(S^{i})}{\left\vert S^{i}\right\vert }$ for some coalition $%
S^{i}\subseteq N$ such that $i\in S^{i}$, for all $i\in N.$

\item[(MD)] Maximal compensation. $\varphi _{0}(v)\leq \sum_{i\in N}\underset%
{S\subseteq N:i\in S}{\min }\frac{v(S_{0})-v(S)}{\left\vert S\right\vert }$
\end{description}

\noindent \textit{Efficiency }implies that the overall benefit is
distributed among the players, while \textit{distributors reduction }%
property provides for a reduction for the distributor $i$ with respect to 
altruistic distribution. This reduction is proportional to the farmer's marginal contribution to coalitions of distributors that
include $i$, specifically $(v(S_{0}^{i})-v(S^{i}))/\left\vert
S^{i}\right\vert $, where $S^{i}$ represents some coalition of distributors.
This property is a modification of the \textit{retailers reduction} property
described in Guardiola et al. (2007). Finally, \textit{maximal compensation }property
states that a farmer can only be compensated up to the sum of their minimal
proportional part of marginal contributions to the sets of distributors.
\medskip

The following theorem states that there exists a unique allocation for
MDF-games that satisfies the EF, DR, and MD properties.

\begin{theorem}
\label{theo}Let $(N_{0},Q,C,b,T,P,\overline{b})$ be an MDF-situationwith SC and NDH, $(N_{0},v)$ be the corresponding
MDF-game and $q^{N_{0}}$ an optimal solution. The unique solution $\sigma $
on the class of MDF-games, satisfying EF, DR and MD is $\sigma (v)=\left(
\sigma _{i}(v)\right) _{i\in N_{0}}$ defined by 
\begin{equation*}
\sigma _{i}(v):=\left\{ 
\begin{array}{ll}
\Lambda _{i}^{N}\left( q^{N_{0}}\right) -\beta _{i}, & i\in N, \\ 
&  \\ 
\sum_{i\in N}\beta _{i}, & i=0,%
\end{array}%
\right.
\end{equation*}

where $\beta _{i}:=\underset{i\in S\subseteq N}{\min }\left\{ \frac{%
\sum_{k\in S}\Lambda _{k}^{S}\left( q^{S_{0}}\right) -\sum_{k\in S}\Pi
_{k}^{S}\left( q^{S}\right) }{\left\vert S\right\vert }\right\} $ for all $%
i\in N.$
\end{theorem}
\begin{proof}
It is clear that $\sigma (v)$ satisfies EF, DR and MD.

To show the converse, take a solution $\varphi $ on the class of MDF-games
that satisfies EF, DR and MD. By DR, $\varphi _{i}(v)=x_{i}^{a}(v)-\alpha
_{i}$ with $\alpha _{i}=(v(S_{0}^{i})-v(S^{i}))/\left\vert S^{i}\right\vert $
for some coalition $S^{i}\subseteq N$ with $i\in S^{i}$, for all distributor 
$i$. According to EF\ $\varphi _{0}(v)=\sum_{i\in N}\alpha _{i}.$ By MC, $%
\sum_{i\in N}\alpha _{i}\leq \sum_{i\in N}\underset{S\subseteq N:i\in S}{%
\min }\frac{v(S_{0})-v(S)}{\left\vert S\right\vert }$. Hence, $\alpha _{i}=%
\underset{S\subseteq N:i\in S}{\min }\frac{v(S_{0})-v(S)}{\left\vert
S\right\vert }$ for all $i\in N.$ We conclude\ $\varphi =\sigma .$\hfill
\end{proof}
\medskip

We will refer to this distribution as the Farmer Compensation
allocation (henceforth, FC-allocation), as each distributor
compensates the farmer for the increase in profits resulting from
cooperation in coalitions which the distributor participates. 
\medskip

Next we prove that properties used in Theorem \ref{theo} are logically
independent.

\begin{example}
(EF fails) Consider $\varphi $ on MDF-games $(N_{0},v)$ defined by 
\begin{equation*}
\varphi _{i}(v):=\left\{ 
\begin{array}{ll}
\Lambda _{i}^{N}\left( q^{N_{0}}\right) -\beta _{i}, & i\in N, \\ 
&  \\ 
0, & i=0,%
\end{array}%
\right.
\end{equation*}%
where $\beta _{i}:=\underset{i\in S\subseteq N}{\min }\left\{ \frac{%
\sum_{k\in S}\Lambda _{k}^{S}\left( q^{S_{0}}\right) -\sum_{k\in S}\Pi
_{k}^{S}\left( q^{S}\right) }{\left\vert S\right\vert }\right\} $ for all $%
i\in N.$ $\varphi (v)$ satisfies DR, MD, but not EF.
\end{example}

\begin{example}
(DR fails) Consider $\varphi $ on MDF-games $(N_{0},v)$ given by 
\begin{equation*}
\varphi _{i}(v):=\left\{ 
\begin{array}{ll}
\Lambda _{i}^{N}\left( q^{N_{0}}\right) -\beta _{i}-1, & i\in N, \\ 
&  \\ 
\sum_{i\in N}\left( \beta _{i}-1\right) , & i=0,%
\end{array}%
\right.
\end{equation*}%
where $\beta _{i}:=\underset{i\in S\subseteq N}{\min }\left\{ \frac{%
\sum_{k\in S}\Lambda _{k}^{S}\left( q^{S_{0}}\right) -\sum_{k\in S}\Pi
_{k}^{S}\left( q^{S}\right) }{\left\vert S\right\vert }\right\} $ for all $%
i\in N.$ $\varphi (v)$ satisfies EF, MD but not DR.
\end{example}

\begin{example}
(MD fails) Let $\varphi $ on MDF-games $(N_{0},v)$ defined by 
\begin{equation*}
\varphi _{i}(v):=\left\{ 
\begin{array}{ll}
\Lambda _{i}^{N}\left( q^{N_{0}}\right) -\beta _{i}^{\ast }, & i\in N, \\ 
&  \\ 
\sum_{i\in N}\beta _{i}^{\ast }, & i=0,%
\end{array}%
\right. 
\end{equation*}%
where $\beta _{i}^{\ast }:=\underset{i\in S\subseteq N}{\max }\left\{ \frac{%
\sum_{k\in S}\Lambda _{k}^{S}\left( q^{S_{0}}\right) -\sum_{k\in S}\Pi
_{k}^{S}\left( q^{S}\right) }{\left\vert S\right\vert }\right\} $ for all $%
i\in N.$ $\varphi (v)$ satisfies EF, DR but not MD.
\end{example}

In the following example we go back to the Example \ref{ejemplo 1}. and compare 
the altruistic with the FC-allocation. 

\begin{example}
Consider again the data in Example \ref{ejemplo 1}. {We compare the
FC-allocation with the altruistic allocation
in Table \ref{table 2}.}
\begin{table}[!h]
\begin{center}
\begin{tabular}{c|c}
$x^{a}(v)$ & $\sigma (v)$ \\ \hline
&   \\ 
$\left( 
\begin{array}{c}
$0$ \\ 
$3195.39$ \\ 
$2093.53$%
\end{array}%
\right) $ & $\left( 
\begin{array}{c}
$4078.4$ \\ 
$743.85$ \\ 
$466.67$%
\end{array}%
\right) $%
\end{tabular}%
\caption{Altruistic vs FC-allocation for Example \ref{ejemplo 1}\label{table 2}}
\end{center}%
\end{table}

We notice that the distributors' compensations are 

\begin{eqnarray*}
\beta _{1}:=\min \left\{ 2957.89-178.57,\frac{5288.92-385.85}{2}\right\} =\min \left\{ 2779.32,2451.54\right\} =2451.54; \\ 
\beta_{2}:=\min \left\{ 1731.03-104.17,\frac{5288.92-385.85}{2}\right\} =\min \left\{ 1626.86,2451.54\right\} =1626.86.
\end{eqnarray*}
The FC-allocation is stable, in the sense of the core, 
and  provides the farmer with a total revenue revenue of $r\left( N_{0}\right)-C+4078.4=3490.38$. The farmer is satisfied with the total revenue obtained from the cooperation with all distributors.
\end{example}

The following theorem shows that FC-allocation is always stable in the sense of
the core.

\begin{theorem}
Let $(N_{0},Q,C,b,T,P,\overline{b})$ be an MDF-situation with SC and NDH, and $(N_{0},v)$ be the corresponding
MDF-game. Then, $\sigma (v)\in Core(N_{0},v)$.
\end{theorem}
\begin{proof}
To prove that $\sigma (v)\in Core(N_{0},v)$ we need to satisfy three
conditions: (1) $\sum_{i\in N_{0}}\sigma _{i}(v)=v(N_{0}),$ (2) $\sum_{i\in
S_{0}}\sigma _{i}(v)\geq v(S_{0})$ and (3) $\sum_{i\in S}\sigma _{i}(v)\geq
v(S)$ for all $S\subset N.$

It easy to check that $\sigma (v)$ satisfy (1) :

\begin{eqnarray*}
\sum_{i\in N_{0}}\sigma _{i}(v) &=&\sigma _{0}(v)+\sum_{i\in N}\sigma
(v)=\sum_{i\in N}\beta _{i}+\sum_{i\in N}\left( \Lambda _{i}^{N}\left(
q^{N_{0}}\right) -\beta _{i}\right) \\
&=&\sum_{i\in N}\beta _{i}+\sum_{i\in N}\Lambda _{i}^{N}\left(
q^{N_{0}}\right) -\sum_{i\in N}\beta _{i}=\sum_{i\in N}\Lambda
_{i}^{N}\left( q^{N_{0}}\right) =v(N_{0})
\end{eqnarray*}

Take $\emptyset \neq S\varsubsetneq N$, then (2) it is satisfy since:

\begin{eqnarray*}
\sum_{i\in S_{0}}\sigma _{i}(v) &=&\sigma _{0}(v)+\sum_{i\in S}\sigma
_{i}(v)=\sum_{i\in N}\beta _{i}+\sum_{i\in S}\left( \Lambda _{i}^{N}\left(
q^{N_{0}}\right) -\beta _{i}\right) \\
&=&\sum_{i\in N}\beta _{i}+\sum_{i\in S}\Lambda _{i}^{N}\left(
q^{N_{0}}\right) -\sum_{i\in S}\beta _{i}=\sum_{i\in S}\Lambda
_{i}^{N}\left( q^{N_{0}}\right) +\sum_{i\in N\backslash S}\beta _{i} \\
&\geq &\sum_{i\in S}\Lambda _{i}^{N}\left( q^{N_{0}}\right) \geq \sum_{i\in
S}\Lambda _{i}^{S}\left( q^{S_{0}}\right) =v(S_{0})
\end{eqnarray*}

and also (3):

\begin{eqnarray*}
\sum_{i\in S}\sigma _{i}(v) &=&\sum_{i\in S}\left( \Lambda _{i}^{N}\left(
q^{N_{0}}\right) -\beta _{i}\right) =\sum_{i\in S}\Lambda _{i}^{N}\left(
q^{N_{0}}\right) -\sum_{i\in S}\beta _{i} \\
&=&\sum_{i\in S}\Lambda _{i}^{N}\left( q^{N_{0}}\right) -\sum_{i\in S}%
\underset{i\in R\subseteq N}{\min }\left\{ \frac{\sum_{k\in R}\Lambda
_{k}^{S}\left( q^{S_{0}}\right) -\sum_{k\in R}\Pi _{k}^{R}\left(
q^{R}\right) }{\left\vert R\right\vert }\right\} \\
&\geq &\sum_{i\in S}\Lambda _{i}^{N}\left( q^{N_{0}}\right) -\sum_{i\in S}%
\frac{\sum_{k\in S}\Lambda _{k}^{S}\left( q^{S_{0}}\right) -\sum_{k\in S}\Pi
_{k}^{S}\left( q^{S}\right) }{\left\vert S\right\vert } \\
&=&\sum_{i\in S}\Lambda _{i}^{N}\left( q^{N_{0}}\right) -\left\vert
S\right\vert \frac{\sum_{k\in S}\Lambda _{k}^{S}\left( q^{S_{0}}\right)
-\sum_{k\in S}\Pi _{k}^{S}\left( q^{S}\right) }{\left\vert S\right\vert } \\
&\geq&\sum_{i\in S}\Pi _{i}^{S}\left( q^{S}\right) =v(S).
\end{eqnarray*}\hfill
\end{proof}

The next example illustrates an MDF-situation consisting of one
farmer and three distributors, where the farmer is not satisfied 
with the revenue obtained from the total cooperation with the distributors.

\begin{example}
\label{ejemplo 3}Consider an MDF-situation with a farmer and 3 distributor where: $Q=8000,$ $C=4000,b(q)=3-\frac{q}{4000}$ 
with $q\in \lbrack 0,6000]$ and $b(q)=1.5$ with $q\in (6000,8000],$ 
$\overline{b}=0.2$ 
\begin{equation*}
p_{1}(q)=\left\{ 
\begin{array}{cc}
8-\frac{4q}{2500}, & 0\leq q\leq 2500, \\
4, & q>2500.%
\end{array}%
\right. \ \ \ ;\ \ t_{1}(q)=\left\{ 
\begin{array}{cc}
1-\frac{0.5q}{2500}, & 0\leq q\leq 2500, \\ 
0.5, & q>2500.%
\end{array}%
\right.
\end{equation*}%
\begin{equation*}
p_{2}(q)=\left\{ 
\begin{array}{cc}
9.85, & 0\leq q\leq 100, \\
5-\frac{1.5q}{1000} + \frac{50}{\sqrt{q}}, & 100< q\leq 1850, \\
3.38, & q>1850.
\end{array}%
\right. \text{ \ \ };\text{\ \ }t_{2}(q)=\left\{ 
\begin{array}{cc}
2-\frac{0.5q}{1000}, & 0\leq q\leq 2000, \\ 
1, & q>2000.%
\end{array}%
\right.
\end{equation*}
\begin{equation*}
p_{3}(q)=\left\{ 
\begin{array}{cc}
9-\frac{4.5q}{1000}, & 0\leq q\leq 1000, \\ 
\frac{4500}{q}, & q>1000.%
\end{array}%
\right. \text{ \ \ };\text{\ \ }t_{3}(q)=\left\{ 
\begin{array}{cc}
1-\frac{0.5q}{1000}, & 0\leq q\leq 1000, \\ 
0.5, & q>1000.%
\end{array}%
\right.
\end{equation*}

Solving the corresponding optimization problems for the different
coalitions, we obtain the optimal orders, the farmer's profit and the characteristic 
function of the game as shown in Table \ref{table 3}.

\begin{table}[!h]
\begin{center}
\begin{tabular}{|c|c|c|c|c|c|c|}
$S$ & $q_{1}^{S}$ & $q_{2}^{S}$ & $q_{3}^{S}$ & $r\left( S\right)$ & $r\left( S\right)-C$ & $v(S)$ \\
\hline
$\{1\}$ & $1739.13$  & - &  - & $4461.24$ & $-3538.76$ & $3478.26$ \\ \hline
$\{2\}$ & -  & $652.47$ & - & $1851$  & $-6149$ & $957.88$ \\ \hline
$\{3\}$ & - & - & $666.66$ & $1888.88$ & $-6111.12$ & $1666.66$ \\ \hline
$\{0,1\}$ & $2392.85$ & - & - & $2317.85$ & $-5682.15$ & $6416.07$\\ \hline
$\{0,2\}$ & - & $1657.07$ & - & $2097.12$ & $-5902.88$ & $2163.56$ \\ \hline
$\{0,3\}$ & - & - & $962.50$ & $1888.75$ & $-6111.25$ & $2105.62$\\ \hline
$\{1,2\}$ & $1989.99$  & $1153.95$ & - & $6960.74$  & $-1039.26$ & $5253.85$\\ \hline
$\{1,3\}$ & $1911.76$ & - & $794.11$  & $6287.19$ & $-1712.81$ & $5808.82$\\ \hline
$\{2,3\}$ & -  & $821.85$ & $721.45$ & $4034.48$ & $-3965.52$ & $2878.69$\\ \hline
$\{0,1,2\}$ & $2392.85$ & $1657.07$ & - & $2814.97$ & $-5185.03$ & $10179.63$\\ \hline
$\{0,1,3\}$ & $2392.85$ & - & $962.50$ & $2606.60$ & $-5393.40$ & $10121.69$ \\ \hline
$\{0,2,3\}$ & - &  $1657.07$ & $962.50$ & $2385.87$ & $-5614.13$ & $5869.18$ \\ \hline
$\{1,2,3\}$ & $2262.69$ & $1491.44$ & $916.94$ & $8558.49$ & $+558.49$ & $8265.96$\\ \hline
$\{0,1,2,3\}$ & $2392.85$ &  $1657.07$ & $962.50$ & $3103.72$ & $-4896.28$ & $13885.26$ \\ \hline
\end{tabular}%
\caption{MDF-game for Example \ref{ejemplo 3}\label{table 3}}
\end{center}
\end{table}

Notice that it satisfies NDH since $q_{S}^{S},q_{S}^{S_{0}}<Q$, for all $S\subseteq N.$
Moreover, it does SC as $\overline{b}=0.2\leq \underset{S\subseteq N}{\min }\left\{ \left(
b(q_{S}^{S})-\frac{C}{Q}\right) \cdot \frac{q_{S}^{S}}{Q-q_{S}^{S}}\right\}
=\min \{0.573,0.207,0.212,1.109,0.932,0.505,1.869\}$.

We now compare the FC-allocation with the altruistic allocation in Table \ref{table 4}, where the distributors' compensations are: 
\begin{eqnarray*}
\beta _{1}:=\min \left\{ 2937.81,2462.89,2156.44,1873.10\right\}=1873.10; \\
\beta _{2}:=\min \left\{ 1205.68,2462.89,1495.25,1873.10\right\}=1205.68; \\
\beta _{3}:=\min \left\{ 438.96,2156.44,1495.25,1873.10\right\}=438.96.
\end{eqnarray*}

\begin{table}[!h]
\begin{center}
\begin{tabular}{c|c}
$x^{a}(v)$ & $\sigma (v)$ \\ \hline
&   \\ 
$\left( 
\begin{array}{c}
$0$ \\ 
$7252.25$ \\ 
$3234.62$ \\
$3398.39$%
\end{array}%
\right) $ & $\left( 
\begin{array}{c}
$3517.74$ \\
$5379.15$ \\ 
$2028.94$ \\
$2959.43$%
\end{array}%
\right) $%
\end{tabular}%
\caption{Altruistic vs FC-allocation for Example \ref{ejemplo 3} \label{table 4}}
\end{center}%
\end{table}

The FC-allocation provides the farmer with a total revenue of $r\left( N_{0}\right)-C+3517.74=-1378.54$. Unfortunately, in this situation the farmer is not satisfied with the (negative) revenue of $r\left( N\right)-C=-4896.28$ obtained from the cooperation with all the distributors. If he were to stay out of this cooperation, he would get a (positive) revenue of  $r\left( N\right)-C=+558.49$.
\end{example}

This unsatisfactory situation for the farmer may be due to the fact that he is a passive agent in MDF-games, as only the profit from sales to customers in the corresponding markets is analyzed. This can lead to situations where the farmer may not be interested in cooperating with distributors, even if he receives compensation from  each of them through the FC-allocation.

\section{Incentives for the farmer to cooperate}

To complete our study of cooperation and compensation, we explore the conditions that can lead to dissatisfaction on the part of the farmer when cooperating. We then focus on designing an allocation that properly incentivizes and rewards the farmer for cooperating. These conditions could be seen as set by the agricultural cooperative or as part of strategic behavior on the part of the distributors to maintain and consolidate business relationships in the short and medium term. 

We first focus on determining when an MDF-situation satisfies $\underset{S\subseteq N}{\max }\{r(S),r(S_{0})\}=r(N_{0})$. This means that the farmer would always be satisfied cooperating with all distributors as he achieves his maximum revenue.  The next proposition characterizes these types of situations in terms of the compensation parameter $\overline{b}.$

\begin{proposition}
Let $(N_{0},Q,C,b,T,P,\overline{b})$ be an MDF-situation with SC and NDH, and $(N_{0},v)$ be the
corresponding MDF-game. Then, $\underset{S\subseteq N}{\max }%
\{r(S),r(S_{0})\}=r(N_{0})$ if and only if 
\begin{equation*}
\frac{1}{(Q-q_{N}^{N_{0}})}\cdot \left( \underset{S\subseteq N}{\max }%
\left\{ b(q_{S}^{S})\cdot q_{S}^{S}\right\} -\frac{C}{Q}\cdot
q_{N}^{N_{0}}\right) \leq \overline{b}\leq \frac{C}{Q}.
\end{equation*}
\end{proposition}
\begin{proof}
We may observe that $\underset{S\subseteq N}{\max }%
\{r(S),r(S_{0})\}=r(N_{0})$ is equivalent to $r(N_{0})\geq r(S)$ and $r(N_{0})\geq r(S_{0})$ for all $S\subseteq N$. Then,
Take a coalition $S\subseteq N.$ We first prove under which conditions it is
satisfied that $r(N_{0})\geq r(S):$%
\begin{eqnarray*}
\frac{C}{Q}\cdot q_{N}^{N_{0}}+\overline{b}\cdot (Q-q_{N}^{N_{0}}) &\geq
&b(q_{S}^{S})\cdot q_{S}^{S}; \\
\overline{b}\cdot (Q-q_{N}^{N_{0}}) &\geq &b(q_{S}^{S})\cdot q_{S}^{S}-\frac{%
C}{Q}\cdot q_{N}^{N_{0}}; \\
\overline{b} &\geq &\frac{1}{(Q-q_{N}^{N_{0}})}\left( b(q_{S}^{S})\cdot
q_{S}^{S}-\frac{C}{Q}\cdot q_{N}^{N_{0}}\right) .
\end{eqnarray*}

Furthermore, if we also want $r(N_{0})\geq r(S_{0})$, the following
condition must be satisfied:%
\begin{eqnarray*}
\frac{C}{Q}\cdot q_{N}^{N_{0}}+\overline{b}\cdot (Q-q_{N}^{N_{0}}) &\geq &%
\frac{C}{Q}\cdot q_{S}^{S_{0}}+\overline{b}\cdot (Q-q_{S}^{S_{0}}); \\
\frac{C}{Q}\cdot q_{N}^{N_{0}}+\overline{b}\cdot (Q-q_{N}^{N_{0}}) &\geq &%
\frac{C}{Q}\cdot q_{S}^{N_{0}}+\overline{b}\cdot (Q-q_{S}^{N_{0}}); \\
\frac{C}{Q}\cdot q_{N\backslash S}^{N_{0}}-\overline{b}\cdot q_{N\backslash
S}^{N_{0}} &\geq &0; \\
\overline{b} &\leq &\frac{C}{Q}.
\end{eqnarray*}

In the first equivalence, we have used property (vi) from Lemma \ref{lemma}.\hfill
\end{proof}
\medskip 

The condition of the above proposition states that the compensation cost is at least equal to the increase in income per kilo unsold when the farmer cooperates with all distributors, and at most equal to the unit cost price of the harvest. As a result, the farmer would obtain the highest possible benefit, which is $r(N_0)-C$. The farmer would have incentives to cooperate with all distributors, and the FC-allocation would be satisfactory for both the farmer and the distributors.

Secondly, we concentrate on those MDF-situations that satisfy either $\overline{b}<\frac{1}{(Q-q_{N}^{N_{0}})}\cdot \left( \underset{S\subseteq N}{\max }\left\{ b(q_{S}^{S})\cdot
q_{S}^{S}\right\} -\frac{C}{Q}\cdot q_{N}^{N_{0}}\right) $ or $\overline{b}>\frac{C}{Q}.$ That is, those situations where the farmer's revenue, when cooperating with all distributors, is lower than the maximum revenue obtained with any group of distributors: $r(N_{0})<\underset{S\subseteq N}{\max }\{r(S),r(S_{0})\}.$ Therefore, we define a new allocation also based on the altruistic
allocation. 

\begin{definition}
Let $(N_{0},Q,C,b,T,P,\overline{b})$ be an MDF-situation with SC and NDH. We define the allocation $%
\theta (v)=\left( \theta _{i}(v)\right) _{i\in N_{0}}$ as follows:
\begin{equation*}
\theta _{i}(v):=\left\{ 
\begin{array}{ll}
\Lambda _{i}^{N}\left( q^{N_{0}}\right) -\alpha _{i}, & i\in N, \\ 
&  \\ 
\sum_{i\in N}\alpha _{i}, & i=0,%
\end{array}%
\right. 
\end{equation*}%
where $\alpha _{i}:=\frac{q_{i}^{N_{0}}}{q_{N}^{N_{0}}}\cdot \underset{%
S\subseteq N}{\max }\{r(S),r(S_{0})\}-\left( \frac{C}{Q}-\overline{b}\right)
\cdot q_{i}^{N_{0}}-\overline{b}\cdot Q\cdot \frac{q_{i}^{N_{0}}}{%
q_{N}^{N_{0}}}.$
\end{definition}

Note that $\alpha _{i}:=\frac{q_{i}^{N_{0}}}{q_{N}^{N_{0}}}\cdot \left( \underset{S\subseteq N}{\max }\{r(S),r(S_{0})\} - r(N_{0}) \right)$, and so if  $\underset{S\subseteq N}{\max }%
\{r(S),r(S_{0})\}=r(N_{0}),$ then $\alpha _{i}=0$ for all $i\in N.$ Hence, $%
\theta (v)=x^{a}(v).$ 

\bigskip
The following proposition shows that the previous allocation always compensates the farmer in such a way that his revenue from collaborating with all distributors is equal to the highest revenue he could achieve with any group of distributors.  

\begin{proposition}
Let $(N_{0},Q,C,b,T,P,\overline{b})$ be an MDF-situation with SC and NDH, and $(N_{0},v)$ be the
corresponding MDF-game. Then, $r(N_{0})+\theta _{0}(v)=\underset{S\subseteq N%
}{\max }\{r(S),r(S_{0})\}.$
\end{proposition}
\begin{proof}
\begin{eqnarray*}
r(N_{0})+\theta _{0}(v) &=&\frac{C}{Q}\cdot q_{N}^{N_{0}}+\overline{b}\cdot
(Q-q_{N}^{N_{0}})+\sum_{i\in N}\alpha _{i} \\
&=&\frac{C}{Q}\cdot q_{N}^{N_{0}}+\overline{b}\cdot (Q-q_{N}^{N_{0}}) \\
&&+\frac{q_{N}^{N_{0}}}{q_{N}^{N_{0}}}\cdot \underset{S\subseteq N}{\max }%
\{r(S),r(S_{0})\}-\left( \frac{C}{Q}-\overline{b}\right) \cdot q_{N}^{N_{0}}-%
\overline{b}\cdot Q\cdot \frac{q_{N}^{N_{0}}}{q_{N}^{N_{0}}} \\
&=&\underset{S\subseteq N}{\max }\{r(S),r(S_{0})\}.
\end{eqnarray*}\hfill
\end{proof}
\medskip 

Notice that the allocation $\theta (v)$ compensates the farmer with the minimum amount necessary for him to feel satisfied with the cooperation with all distributors. With this allocation rule, the distributors provide compensation proportional to their order sizes in the grand coalition ($N_{0})$. Therefore, from now on, we will refer to this distribution
as the minimal proportional compensation allocation (henceforth, MPC-allocation). 

In Table \ref{table 5} we compare the altruistic allocations and the FC-allocation  with MDF-allocation for Example \ref{ejemplo 1}.

\begin{table}[!h]
\begin{center}
\begin{tabular}{c|c|c}
$x^{a}(v)$ & $\sigma (v)$ & $\theta (v)$\\ \hline
&   \\ 
$\left( 
\begin{array}{c}
$0$ \\ 
$3195.39$ \\ 
$2093.53$%
\end{array}%
\right) $ & $\left( 
\begin{array}{c}
$4078.4$ \\ 
$743.85$ \\ 
$466.67$%
\end{array}%
\right) $%
&$\left( 
\begin{array}{c}
$1148.78$ \\ 
$2501.34$ \\ 
$1638.81$%
\end{array}%
\right) $%
\end{tabular}%
\caption{Comparison of all allocations for Example \ref{ejemplo 1}\label{table 5}}
\end{center}%
\end{table}

As we can see, the MPC-allocation also belongs to the core and satisfies the farmer, $r(N_{0})+1148.78=\underset{S\subseteq N%
}{\max }\{r(S),r(S_{0})\}=r(\{1,2\})=3560.75$, with a minimum redcution of profits for the distributors. In this example, the farmer would prefer an allocation according to FC-allocation. 

On the other hand, in Table \ref{table 6} for Example \ref{ejemplo 3}, we have that  $r(N_{0})+5454.77=\underset{S\subseteq N%
}{\max }\{r(S),r(S_{0})\}=r(\{1,2,3\})=8558.49$. Moreover, $\theta (v)$ also belongs to the core.

\begin{table}[!h]
\begin{center}
\begin{tabular}{c|c|c}
$x^{a}(v)$ & $\sigma (v)$ & $\theta (v)$ \\ \hline
&   \\ 
$\left( 
\begin{array}{c}
$0$ \\ 
$7252.25$ \\ 
$3234.62$ \\
$3398.39$%
\end{array}%
\right) $ & $\left( 
\begin{array}{c}
$3517.74$ \\
$5379.15$ \\ 
$2028.94$ \\
$2959.43$%
\end{array}%
\right) $%
& $\left( 
\begin{array}{c}
$5454.77$ \\
$4588.14$ \\ 
$1528.02$ \\
$2314.33$%
\end{array}%
\right) $%
\end{tabular}%
\caption{Comparison of all allocations for Example \ref{ejemplo 3}\label{table 6}}
\end{center}%
\end{table} 

Although the compensation with the MPC-allocation is always satisfactory for the farmer, unfortunately it does not always lead to a stable allocation in terms of the core. The next example shows that the MPC-allocation is not a core allocation in general.

\begin{example}
\label{ejemplo 5}Consider a MDF-situation with 2 distributor and 1
farmer where: $Q=8000,$ $C=4000,b(q)=3-\frac{q}{4000}$ with $q\in \lbrack
0,6000]$ and $b(q)=1.5$ with $q\in (6000,8000],$ $\overline{b}=0.01$ 
\begin{equation*}
p_{1}(q)=\left\{ 
\begin{array}{cc}
8-\frac{5q}{3000}, & 0\leq q\leq 2400, \\ 
4, & q>2400.%
\end{array}%
\right. \ \ ;\ \ t_{1}(q)=\left\{ 
\begin{array}{cc}
1-\frac{q}{5000}, & 0\leq q\leq 2500, \\ 
0.5, & q>2500.%
\end{array}%
\right. 
\end{equation*}%
\begin{equation*}
p_{2}(q)=\left\{ 
\begin{array}{cc}
3-\frac{3q}{2000}+\frac{20}{\sqrt{2x}}, & 0\leq q\leq 1000, \\ 
\frac{1}{10}\left( 2\sqrt{5}+15\right) , & q>1000.%
\end{array}%
\right. \text{ \ \ };\text{\ \ }t_{2}(q)=\left\{ 
\begin{array}{cc}
1-\frac{q}{2000}, & 0\leq q\leq 1000, \\ 
0.5, & q>1000.%
\end{array}%
\right. 
\end{equation*}

Solving the corresponding optimization problems for the different
coalitions, we obtain the optimal orders, the farmer's revenue and the characteristic function as shown in Table \ref{table 7}.

\begin{table}[!h]
\begin{center}
\begin{tabular}{|c|c|c|c|c|c|}
$S$ & $q_{1}^{S}$ & $q_{2}^{S}$ & $r\left( S\right)$ & $r\left( S\right)-C$ & $v(S)$ \\ 
\hline
$\{1\}$ & $1643.84$ & - & $4255.96$ & $+255.96$ & $3287.67$ \\ \hline
$\{2\}$ & - & $44$ & $131.52$ & $-3868.48$ & $48.36$ \\ \hline
$\{0,1\}$ & $2219.32$ & - & $1167.47$ & $-2832.53$ & $7143.88$ \\ \hline
$\{0,2\}$ & - & $874.55$ & $508.53$ & $-3491.47$ & $893.96$ \\ \hline
$\{1,2\}$ & $1688.32$ & $216.49$ & $4807.36$ & $+807.36$ & $3424.46$ \\ \hline
$\{0,1,2\}$ & $2219.32$ & $874.55$ & $1596$ & $-2404$ & $8117.84$ \\ \hline
\end{tabular}%
\caption{MDF-game for Example \ref{ejemplo 5} \label{table 7}}
\end{center}
\end{table}

Notice that it satisfies NDH since $q_{S}^{S},q_{S}^{S_{0}}<Q$ for all $S\subseteq N.$
Moreover, it does SC as $\overline{b}=0.01<\underset{S\subseteq N}{\min }\left\{ \left(
b(q_{S}^{S})-\frac{C}{Q}\right) \cdot \frac{q_{S}^{S}}{Q-q_{S}^{S}}\right\}
=0.013.$ However, the MPC-allocations, $\theta (v)=\left( 3211.36,4862.9,43.58\right)$
does not belong to the core since $\theta _{2}(v)=43.58<48.36=v(\{2\}).$

\end{example}

The last theorem provides a necessary and sufficient condition for the MCP-allocation to belong to the core.

\begin{theorem}
\label{theo2}Let $(N_{0},Q,C,b,T,P,\overline{b})$ be an MDF-situation
with SC and NDH, $(N_{0},v)$
be the corresponding MDF-game. Then $\theta (v)\in Core(N,v)$ if and only if 
\begin{equation*}
\underset{S\subseteq N}{\max }\{r(S),r(S_{0})\}\leq \frac{q_{N}^{N_{0}}}{%
q_{S}^{N_{0}}}\left( \sum_{i\in S}\left[\left(
p_{i}(q_{i}^{N_{0}})-t_{i}(q_{i}^{N_{0}})\right) \cdot
q_{i}^{N_{0}}\right]-v(S)\right),
\end{equation*}%
for all $S\subseteq N.$
\end{theorem}
\begin{proof}
It is easy to check that $\theta (v)$ satisfy EF. Next we prove that $%
\alpha _{i}\geq 0$ for all $i\in N.$ Take a distributor $i\in N,$%
\begin{eqnarray*}
\alpha _{i} &=&\frac{q_{i}^{N_{0}}}{q_{N}^{N_{0}}}\cdot \underset{S\subseteq
N}{\max }\{r(S),r(S_{0})\}-\left( \frac{C}{Q}-\overline{b}\right) \cdot
q_{i}^{N_{0}}-\overline{b}\cdot Q\cdot \frac{q_{i}^{N_{0}}}{q_{N}^{N_{0}}} \\
&\geq &\frac{q_{i}^{N_{0}}}{q_{N}^{N_{0}}}\cdot r(N_{0})-\left( \frac{C}{Q}-%
\overline{b}\right) \cdot q_{i}^{N_{0}}-\overline{b}\cdot Q\cdot \frac{%
q_{i}^{N_{0}}}{q_{N}^{N_{0}}} \\
&=&\frac{C}{Q}\cdot q_{i}^{N_{0}}+\frac{q_{i}^{N_{0}}}{q_{N}^{N_{0}}}\cdot 
\overline{b}\cdot (Q-q_{N}^{N_{0}})-\left( \frac{C}{Q}-\overline{b}\right)
\cdot q_{i}^{N_{0}}-\overline{b}\cdot Q\cdot \frac{q_{i}^{N_{0}}}{%
q_{N}^{N_{0}}} \\
&=&\frac{q_{i}^{N_{0}}}{q_{N}^{N_{0}}}\cdot \overline{b}\cdot Q-\overline{b}%
\cdot Q\cdot \frac{q_{i}^{N_{0}}}{q_{N}^{N_{0}}}=0.
\end{eqnarray*}

Hence, $\sum_{i\in S_{0}}\theta (v)=\sum_{i\in S}\Lambda _{i}^{N}\left(
q^{N_{0}}\right) +\sum_{i\in N\backslash S}\alpha _{i}\geq v(S_{0}).$ Next,
we have to prove that $\sum_{i\in S}\theta (v)\geq v(S)$ for each coalition $%
S\subseteq N.$ With this objective, we have that:  
\begin{eqnarray*}
\sum_{i\in S}\theta (v) &=&\sum_{i\in S}\Lambda _{i}^{N}\left(
q^{N_{0}}\right) -\sum_{i\in S}\alpha _{i} \\
&=&\sum_{i\in S}\Lambda _{i}^{N}\left( q^{N_{0}}\right) -\sum_{i\in S}\left( 
\frac{q_{i}^{N_{0}}}{q_{N}^{N_{0}}}\cdot \underset{S\subseteq N}{\max }%
\{r(S),r(S_{0})\}-\left( \frac{C}{Q}-\overline{b}\right) \cdot q_{i}^{N_{0}}-%
\overline{b}\cdot Q\cdot \frac{q_{i}^{N_{0}}}{q_{N}^{N_{0}}}\right)  \\
&=&\sum_{i\in S}\left[\left( p_{i}(q_{i}^{N_{0}})-t_{i}(q_{i}^{N_{0}})-\frac{C}{Q}-%
\frac{\overline{b}\cdot (Q-q_{N}^{N_{0}})}{q_{N}^{N_{0}}}\right) \cdot
q_{i}^{N_{0}}\right] \\
&&-\sum_{i\in S}\left( \frac{q_{i}^{N_{0}}}{q_{N}^{N_{0}}}\cdot \underset{%
S\subseteq N}{\max }\{r(S),r(S_{0})\}-\left( \frac{C}{Q}-\overline{b}\right)
\cdot q_{i}^{N_{0}}-\overline{b}\cdot Q\cdot \frac{q_{i}^{N_{0}}}{%
q_{N}^{N_{0}}}\right)  \\
&=&\sum_{i\in S}\left[\left( p_{i}(q_{i}^{N_{0}})-t_{i}(q_{i}^{N_{0}})\right)
\cdot q_{i}^{N_{0}}\right]-\left( \frac{C}{Q}-\overline{b}\right) \cdot
q_{S}^{N_{0}}-\overline{b}\cdot Q\cdot \frac{q_{S}^{N_{0}}}{q_{N}^{N_{0}}} \\
&&-\frac{q_{S}^{N_{0}}}{q_{N}^{N_{0}}}\cdot \underset{S\subseteq N}{\max }%
\{r(S),r(S_{0})\}+\left( \frac{C}{Q}-\overline{b}\right) \cdot q_{S}^{N_{0}}+%
\overline{b}\cdot Q\cdot \frac{q_{S}^{N_{0}}}{q_{N}^{N_{0}}} \\
&=&\sum_{i\in S}\left[\left( p_{i}(q_{i}^{N_{0}})-t_{i}(q_{i}^{N_{0}})\right)
\cdot q_{i}^{N_{0}}\right]-\frac{q_{S}^{N_{0}}}{q_{N}^{N_{0}}}\cdot \underset{%
S\subseteq N}{\max }\{r(S),r(S_{0})\}.
\end{eqnarray*}

Finally, $\sum_{i\in S}\theta (v) \geq v(S)$ is equivalent to:
\begin{eqnarray*}
\underset{S\subseteq N}{\max }\{r(S),r(S_{0})\} &\leq &\frac{q_{N}^{N_{0}}}{%
q_{S}^{N_{0}}}\left( \sum_{i\in S}\left[\left(
p_{i}(q_{i}^{N_{0}})-t_{i}(q_{i}^{N_{0}})\right) \cdot
q_{i}^{N_{0}}\right]-v(S)\right) .
\end{eqnarray*}\hfill
\end{proof}


\section{Concluding Remarks}

In this paper, we have extended and applied the model from Guardiola et al. (2007) to the context of agricultural supply chains. We have studied supply chains with a single product and a single period, where multiple distributors source the harvested product from a farmer through the intermediation of an agricultural cooperative. In our model, distributors choose an order size that maximizes their profits. Additionally, the price per kilogram of the product decreases as the total order size increases, making it sensible for distributors to cooperate and place larger orders. Furthermore, distributors can also cooperate with the farmer by securing a reduction in the price per kilogram down to the cost price. In return, the distributors compensate the farmer for the kilograms that are not acquired as part of the cooperation agreement. Additionally, this cooperation arrangement allows the farmer to participate in the profits generated by the distributors from sales to the public across different markets. We have demonstrated that under two initial assumptions (sustainable compensation and no depleting the harvest), the cooperation of all agents generates more profit than the sum of their individual profits.

We have introduced multidistributor-farmer games (MDF-games), which assign to each coalition the optimal benefit they can generate by cooperating with each other. We have proved that these games are superadditive, and also balanced by altruistic allocation, which distributes the total benefit among the distributors. However, despite being a core allocation, the altruistic allocation does not respect the farmer's agreement to participate in the profits from the distributors' sales to the public. For this reason, we have proposed and analyzed an alternative allocation that compensates the farmer. Specifically, each distributor compensates the farmer on the basis of a portion of the marginal contribution that the latter makes to the coalitions in which they participate. We have referred to this distribution as the FC-allocation. Moreover, we have demonstrated that this allocation belongs to the core and we have provided a characterization of it. 

Finally, we have focused on situations in which the farmer may be dissatisfied with the revenue derived from cooperating with all distributors. These situations occur when the revenue generated from the collaboration of all agents is lower than the maximum revenue obtained with any group of distributors. We have proposed an alternative allocation, the MPC-allocation, which compensates the farmer with the minimum amount necessary for him to be satisfied with his cooperation with all distributors. Specifically, distributors compensate the farmer based on the increase in revenue proportional to their optimal order sizes. To conclude our study, we provide a necessary and sufficient condition for the MPC-allocation to be coalitionally stable, meaning it constitutes a core allocation.

There are several lines of future research that we can consider. The fisrt is to examine the same model but incorporate multiple farmers who produce the same product, as well as explore additional forms of cooperation among farmers. This would allow us to analyze the dynamics and outcomes of supply chains with multiple farmers collaborating with distributors. The second is to investigate the implications of the game when the optimal orders of the distributors reach the total quantity of kilograms produced. This scenario would explore the effects of maximum demand on cooperation strategies and outcomes within the supply chain. The third is to consider the case where the demand in each market for each distributor exhibits random behavior. Introducing stochastic elements into the model would provide insights into the robustness and adaptability of cooperative strategies in response to market fluctuations. By exploring these directions, we can deepen our understanding of the complexities and dynamics of the multidistributor-farmer situation, contributing to the advancement of cooperative game theory in agricultural supply chain settings.
\bigskip

\subsection*{Acknowledgements}
This work is part of the R+D+I project grant PID2022-137211NB-100 that was funded by MCIN/AEI/10.13039/50110001133/ and by "ERDF A a way of making EUROPE"/UE. This research was also funded by project PROMETEO/2021/063 from the Comunidad Valenciana.

\bigskip

\end{document}